\documentclass[journal, twoside,onecolumn]{IEEEtran}
\usepackage{}
\ifCLASSINFOpdf
\else
\fi
\usepackage{dsfont}
\usepackage{ifpdf}
\usepackage{graphicx}                                
\usepackage{amsmath,amsfonts,amssymb}                
\usepackage{cases}                                   
\usepackage{stmaryrd}                                
\usepackage{psfrag}
\usepackage{placeins}                                
\usepackage{colortbl}
\usepackage{url}                          
\usepackage{amsthm}                                  
\usepackage{multirow,booktabs}
\usepackage[subfigure]{graphfig}
\usepackage{footnote}
\usepackage{amsfonts}
\usepackage{bm}
 \usepackage{amsfonts}
\usepackage{eucal }
\newtheorem{theorem}{Theorem}
\newtheorem{lemma}[theorem]{Lemma}

\newtheorem{remark}[theorem]{Remark}
\theoremstyle{definition}
\newtheorem{definition}[theorem]{Definition}

\begin{document}
%
\title{Orthogonal subspace based fast iterative thresholding algorithms for joint sparsity recovery}

\author{Ningning Han, Shidong Li, and Jian Lu  \emph{Member, IEEE}
\thanks{This work was supported by the National Natural Science Foundation of China under grants 61972265, 11871348 and 61373087, by the Natural Science Foundation of Guangdong Province of China under grant 2020B1515310008, by the Educational Commission of Guangdong Province of China under grant 2019KZDZX1007, and by the Guangdong Key Laboratory of Intelligent Information Processing, China, and the NSF of USA (DMS-1615288).
}
\thanks{Ningning Han (ningninghan@szu.edu.cn) and Jian Lu ({\em{corresponding author}}, jianlu@szu.edu.cn) are with Shenzhen Key Laboratory of Advanced Machine Learning and Applications, College of Mathematics and Statistics, Shenzhen
University, Shenzhen, 518060.}
\thanks{Shidong Li (shidong@sfsu.edu) is with Department of Mathematics, San Francisco State University, San Francisco, CA94132.}
}

\markboth{}%
{HAN \MakeLowercase{\textit{et al.}}: LaTex Template for IEEE SPL}

\maketitle

\begin{abstract}
Sparse signal recoveries from multiple measurement vectors (MMV) with joint sparsity property have many applications in signal, image, and video processing. The problem becomes much more involved when snapshots of the signal matrix are temporally correlated. With signal's temporal correlation in mind,  we provide a framework of iterative MMV algorithms based on thresholding, functional feedback and null space tuning. Convergence analysis for exact recovery is established. Unlike most of iterative greedy algorithms that select indices in a measurement/solution space, we determine indices based on an orthogonal subspace spanned by the iterative sequence. In addition, a functional feedback that controls the amount of energy relocation from the ``tails'' is implemented and analyzed.  It is seen that the principle of functional feedback is capable to lower the number of iteration and speed up the convergence of the algorithm. Numerical experiments demonstrate that the proposed algorithm has a clearly advantageous balance of efficiency, adaptivity and accuracy compared with other state-of-the-art algorithms.
\end{abstract}


%
\IEEEpeerreviewmaketitle

\section{Introduction}
In sparse reconstruction signal models with joint sparsity property, signals are sampled at $L$ time instances, resulting in the multiple measurement vector (MMV) model:
\begin{equation}\label{equation2}
 \begin{array}{l}
Y=\Phi X+E,
 \end{array}
\end{equation}
where $Y\in\mathbb{C}^{M\times L}$ is the observation matrix containing $L$ measurement/snapshot (column) vectors, $\Phi\in\mathbb{C}^{M\times N}$ is the measurement matrix governed by the specific physical system, and $X\in\mathbb{C}^{N\times L}$ is the underlying source signal matrix, to be recovered.  $E\in\mathbb{C}^{M\times L}$ is an additive measurement noise matrix.

In this system, $L$ measurements share the same row support and elements
in each nonzero row of $X$ are temporally correlated. The solution problem to a noiseless MMV model can be formulated as
\begin{equation}\label{equation3}
 \begin{array}{l}
 \min\limits_{X}\|X\|_{0}~\text{s.t.}~Y=\Phi X,
 \end{array}
\end{equation}
where $\|X\|_{0}=|\text{supp}(X)|$, $\text{supp}(X)=\{1\leq i\leq N: X_{i\cdot}\neq 0\}$, $ X_{i\cdot}$ is the $i$-th row of $X$. In \cite{Chen}, the authors have shown that $X$ is the unique solution of (\ref{equation3}) if
\begin{equation}\label{limit}
 \begin{array}{l}
 \|X\|_{0}<\frac{\text{spark}(\Phi)+\text{rank}(Y)-1}{2},
 \end{array}
\end{equation}
where spark$(\Phi)$ is the smallest number of linearly dependent columns of $\Phi$.

A large majority of effective algorithms for solving (\ref{equation3}) are based on two strategies: extending single measurement vector (SMV) algorithms or exploiting signal subspaces. 
Well-known algorithms of the first class include simultaneous orthogonal matching pursuit (SOMP) \cite{Tropp1}-\cite{Determeand4}, mixed norm minimization techniques \cite{Cotter}-\cite{Khanna}, simultaneous greedy algorithms \cite{Foucart,Blanchard}. However, these algorithms, without exploiting subspace structures or temporal correlations, have not offered realistic improvements over performances than that of SMV cases. Recently, a multiple sparse Bayesian learning (MSBL) algorithm \cite{Zhangqq}-\cite{Shang 2}, as an extension of sparse Bayesian SMV algorithms, is seen to improve recovery performances by modeling temporal correlation of sparse vectors. Another strategy is to exploit subspace structures spanned by measurement vectors. Representative algorithms include, e.g., sequential compressive MUSIC (SeqCS-MUSIC) \cite{MKim,Kimdd}, subspace-augmented MUSIC (SA-MUSIC+OSMP) \cite{KLee}, rank aware order recursive matching pursuit (RA-ORMP) \cite{Davies,Blanchards,junKim}, semi-supervised MUSIC (SS-MUSIC) \cite{ZWen} etc.

In this report, we provide a computationally efficient ``greedy'' algorithm for joint sparsity signal recoveries from their multiple measurement vectors. The proposed algorithm combines procedures of hard thresholding (HT), functional feedback ($f$-FB) for ``tail'' energy shrinkage and enhanced feasibility, the null space tuning (NST), and a novel variable selection mechanism.  The novel criterion of variable selection is based on estimations of significant coefficients in an orthogonal subspace of the iterative sequence.  The cardinality of selected variables is determined by the feedback function $f$. Experimental results show that the proposed algorithm provides superior performances in terms of the efficiency and the critical sparsity (i.e., the maximum sparsity level at which the perfect recovery is guaranteed \cite{WDai}). In fact, the rate of successful recovery of our algorithm has broken through the algebraic upper bound given in (\ref{limit}).
\begin{figure}[th]
\centering
\begin{tabular}{c@{\hskip -0.9cm}c}
\resizebox {1.85in}{1.5in} {\includegraphics{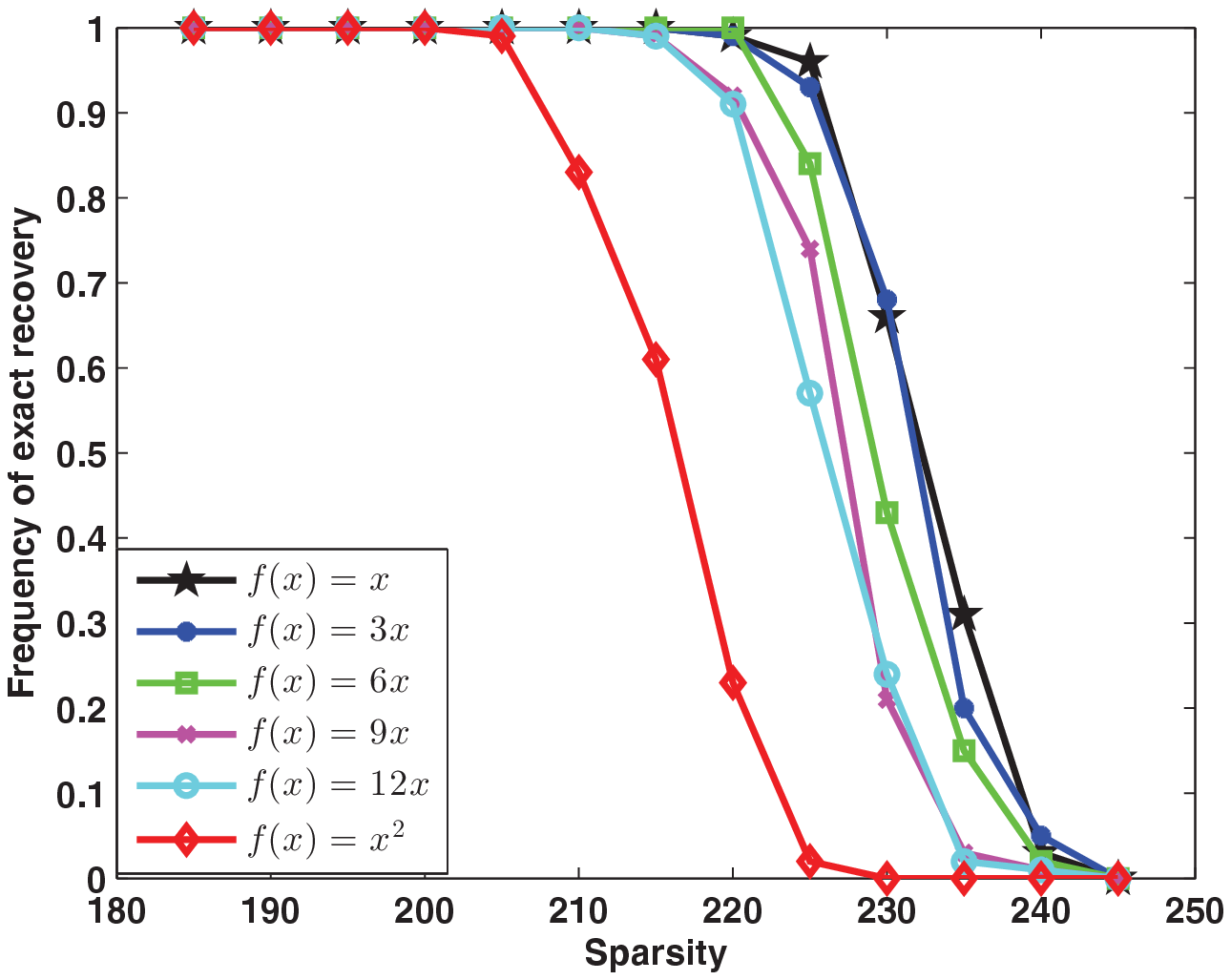}}
\quad & \quad
\resizebox {1.85in}{1.5in} {\includegraphics{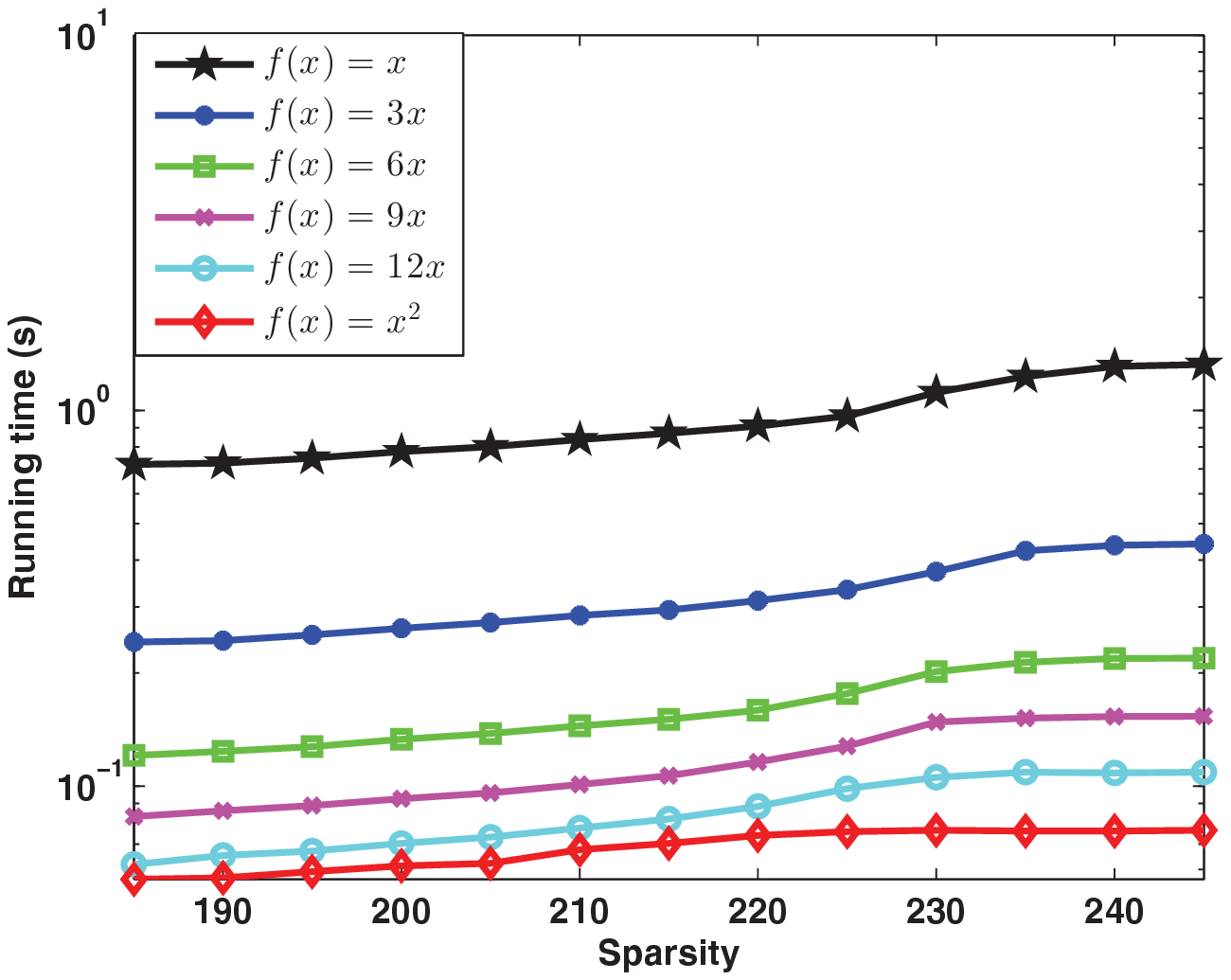}}
\end{tabular}
\vskip -0.3cm
\caption{\small  Left: Frequency of exact recovery as a function of sparsity; right: running time as a function of sparsity.} 
\label{figure:Micro_Doppler}
\end{figure}
\section{Orthogonal subspace NST+HT+$f$-FB algorithm}
\subsection{Notations}
A submatrix of $\Phi$ with columns indexed by a set $I$ is denoted by $\Phi_{I}$ and a submatrix of $\Phi$ with rows indexed by a set $J$ is denoted by $\Phi_{(J)}$. We denote the $i$-th row and the $j$-th column of a matrix $\Phi$ by $\Phi_{i\cdot}$, and $\Phi_{\cdot j}$, respectively.  $T\triangle T'$ is the symmetric difference of $T$ and $T'$, i.e., $T\triangle T' =(T\setminus T')\cup(T'\setminus T)$. $\bm{H}_{T}(X)$ is a linear operator that sets all but elements belong to rows indexed by $T$ of $X$ to zero.

\begin{center}
     \begin{tabular}{lp{76mm}} \toprule[1pt]
     \multicolumn{2}{l}{{\bf Algorithm}~$1$~~OSNST+HT+$f$-FB}\\
       \hline
     &{\bf Input:} $\Phi$, $Y$, $\epsilon$, $f(\cdot)$, $K$;\\
      &{\bf Output:} $W$;\\
      &{\bf Initialize:} $k=1$, $W^{0}=0$;\\
      &{\bf While} $\|Y-\Phi W^{k-1}\|_{2}>\epsilon$ \text{and} $k<K$ {\bf do}\\
      &~~~~~~$X^{k}=W^{k-1}+\Phi^{\ast}(\Phi \Phi^{\ast})^{-1}(Y-\Phi W^{k-1})$; \\
      &~~~~~~$Q^{k}=$orth$(X^{k})$; \\
      &~~~~~~$T_{k}=\{$Indices of $f(k)$ largest $\|Q^{k}_{i\cdot}\|_{2}\}$;\\
      &~~~~~~$W_{T_{k}}^{k}=X_{T_{k}}^{k}+(\Phi_{T_{k}}^{\ast}\Phi_{T_{k}})^{-1}\Phi_{T_{k}}^{\ast}\Phi_{T^{c}_{k}}X_{T^{c}_{k}}^{k}$;\\
      &~~~~~~$W_{T^{c}_{k}}^{k}=0$;\\
      &~~~~~~$k=k+1$;\\
      &{\bf end while};\\
     \hline
     \end{tabular}
\end{center}

\subsection{Algorithm framework}

The iterative framework of approximation and null space tuning (NST) algorithms is as follows
\begin{equation*} \label{eq4}
\left\{ \begin{aligned}
         \begin{aligned}
         &W^{k}=\mathbb{D}(X^{k}),\\
         &X^{k+1}=X^{k}+\mathbb{P}(W^{k}-X^{k}).\\
         \end{aligned}
         \end{aligned} \right.
\end{equation*}
Here $\mathbb{D}(X^{k})$ approximates the desired solution by various principles, and $\mathbb{P}:=I-\Phi^{\ast}(\Phi\Phi^{\ast})^{-1}\Phi$ is the orthogonal projection onto ker$(\Phi)$.

Since the sequence $\{X^{k}\}$ is always feasible (i.e., $Y=\Phi X^{k}$) under the NST principle, one may split $Y$ as
\begin{equation*}\label{eq5}
Y=\Phi X=\Phi_{T_{k}}X_{(T_{k})}^{k}+\Phi_{T^{c}_{k}}X_{(T^{c}_{k})}^{k},
\end{equation*}
where $T_{k}$ includes indices of $f(k)$ largest $\|Q^{k}_{i\cdot}\|_{2}$ ($i\in\{1,\ldots,N\}$),  $f(\cdot)\geq0$ is a non-decreasing function and columns of $Q^{k}$ are an orthonormal basis for the column space of $X^{k}$, i.e., $Q^{k}$=orth$(X^{k})$.
The mechanism of feedback is to feed the contribution of $\Phi_{T^{c}_{k}}X_{(T^{c}_{k})}^{k}$ to $Y$ back to im($\Phi_{T_{k}}$), the image of $\Phi_{T_{k}}$. A straightforward way is to set
\begin{equation*}\label{eq6}
\Lambda^{k}=\arg\min \limits_{\Lambda}\|\Phi_{T_{k}}\Lambda-\Phi_{T^{c}_{k}}X_{(T^{c}_{k})}^{k}\|_{2},
\end{equation*}
which has the best/least-square solution
\begin{equation*}\label{eq7}
\Lambda^{k}=(\Phi_{T_{k}}^{\ast}\Phi_{T_{k}})^{-1}\Phi_{T_{k}}^{\ast}\Phi_{T_{k}^{c}}X_{(T_{k}^{c})}^{k}.
\end{equation*}
The orthogonal subspace iterative thresholding algorithm with functional feedback and null space tunning (OSNST+HT+$f$-FB) is then established in {\em Algorithm 1}.

\subsection{Convergence analysis}
In this paper, we assume the number of snapshots is smaller than the dimension of measurement, i.e., $L<M$, and the measurement matrix $Y$ is full column rank, i.e., rank$(Y)=L$. We now turn to the convergence of OSNST+HT+$f$-FB.
\begin{definition}\label{RIP} \cite{Candes}. For each integer $s = 1, 2,\cdots$, the restricted isometry constant (RIC) $\delta_{s}$ of a matrix $\Phi$ is defined as the smallest number $\delta_{s}$ such that
\begin{equation*}
\begin{array}{l}
(1-\delta_{s})\|X\|_{F}^{2}\leq\|\Phi X\|_{F}^{2}\leq(1+\delta_{s})\|X\|_{F}^{2}
 \end{array}
\end{equation*}
holds for all $s$ row-sparse matrix $X$. Equivalently, it is given by
\begin{equation*}
\begin{array}{l}
\delta_{s}=\max\limits_{|S|\leq s}\|I-\Phi_{S}^{\ast}\Phi_{S}\|_{2}.
 \end{array}
\end{equation*}
\end{definition}

\begin{definition}\label{P-RIP} \cite{SDLi}. For each integer $s = 1, 2,\cdots $ the preconditioned restricted isometry constant $\gamma_{s}$ of a matrix $A$ is defined as the smallest number $\gamma_{s}$ such that
\begin{equation*}
\begin{array}{l}
(1-\gamma_{s})\|X\|_{F}^{2}\leq\|(\Phi\Phi^{\ast})^{-\frac{1}{2}}\Phi X\|_{F}^{2}
 \end{array}
\end{equation*}
holds for all $s$ row-sparse matrix $X$. In fact, the preconditioned restricted isometry constant $\gamma_{s}$ represents the restricted isometry property of the preconditioned matrix $(\Phi\Phi^{\ast})^{-\frac{1}{2}}\Phi$. Since
\begin{equation*}
\begin{array}{l}
\|(\Phi\Phi^{\ast})^{-\frac{1}{2}}\Phi X\|_{F}\leq\|(\Phi\Phi^{\ast})^{-\frac{1}{2}}\Phi\|_{2}\|X\|_{F}=\|X\|_{F},
 \end{array}
\end{equation*}
$\gamma_{s}$ is actually the smallest number such that, for all $s$ row-sparse matrix $X$,
\begin{equation*}
\begin{array}{l}
(1-\gamma_{s})\|X\|_{F}^{2}\leq\|(\Phi\Phi^{\ast})^{-\frac{1}{2}}\Phi X\|_{F}^{2}\leq(1+\gamma_{s})\|X\|_{F}^{2}.
 \end{array}
\end{equation*}
It indicates $\gamma_{s}(\Phi)=\delta_{s}((\Phi\Phi^{\ast})^{-\frac{1}{2}}\Phi)$. Equivalently, it is given by
\begin{equation*}
\begin{array}{l}
\gamma_{s}=\max\limits_{|S|\leq s}\|I-\Phi_{S}^{\ast}(\Phi\Phi^{\ast})^{-1}\Phi_{S}\|_{2}.
 \end{array}
\end{equation*}
\end{definition}

\begin{definition}\label{feasible}
Let the feasible solution space of (\ref{equation3}) be $\mathcal {X}=\{X\in\mathbb{C}^{N\times L}: Y=\Phi X\}$. Define the modified matrix condition number of $\mathcal {X}$ by $\alpha=\max\limits_{X\in\mathcal {X}}\frac{\sigma_{\max}(X)}{\sigma_{\min}(X)}$, where $\sigma_{\min}(X)$ and $\sigma_{\max}(X)$ denote the smallest and the largest nonzero singular values of $X$, respectively.
\end{definition}

\begin{lemma} \label{lemma2.7}
Let  $U,V\in\mathbb{C}^{N\times L}$ with $|supp (U) \cup supp (V)|\leq t$, then $|\langle U,(I-\Phi^{\ast}\Phi)V\rangle|\leq\delta_{t}\|U\|_{F}\|V\|_{F}$. Suppose $|R \cup supp (V)|\leq t$, then $\|[(I-\Phi^{\ast}\Phi)V]_{(R)}\|_{F}\leq\delta_{t}\|V\|_{F}$.
\end{lemma}

\begin{proof}
Let $T=supp (U)\cup  supp (V)$, we then have
\begin{equation*}
 \begin{array}{l}
|\langle U,(I-\Phi^{\ast}\Phi)V\rangle| =|\langle U,V\rangle-\langle \Phi U,\Phi V\rangle| \\
=|\langle U_{(T)},V_{(T)}\rangle-\langle \Phi_{T}U_{(T)},\Phi_{T}V_{(T)}\rangle|\\
=|\langle U_{(T)},(I-\Phi_{T}^{\ast}\Phi_{T})V_{(T)}\rangle|\\
\leq\|U_{(T)}\|_{F}\|(I-\Phi_{T}^{\ast}\Phi_{T})V_{(T)}\|_{F}\\
\leq\|U_{(T)}\|_{F}\|I-\Phi_{T}^{\ast}\Phi_{T}\|_{2}\|V_{(T)}\|_{F}\\
\leq\delta_{t}\|U\|_{F}\|V\|_{F}.\\
 \end{array}
\end{equation*}

The first and the second inequalities are due to the Cauchy-Schwarz inequality, and the sub-multiplicativity of matrix norms, respectively. The last step is by {\em Definition \ref{RIP}}.
It then follows that
\begin{equation*}
 \begin{array}{l}
\|[(I-\Phi^{\ast}\Phi)V]_{(R)}\|_{F}^{2} =\langle(\bm{H}_{R}\left((I-\Phi^{\ast}\Phi)V\right),(I-\Phi^{\ast}\Phi)V\rangle\leq\delta_{t}\|[(I-\Phi^{\ast}\Phi)V]_{(R)}\|_{F}\|V\|_{F}.
 \end{array}
\end{equation*}
Therefore, $\|[(I-\Phi^{\ast}\Phi)V]_{(R)}\|_{F}\leq\delta_{t}\|V\|_{F}$.
\end{proof}

\begin{remark}\label{remark11}
Let $\gamma_t$ be the P-RIP constant of $\Phi$ and $U,V\in\mathbb{C}^{N\times L}$ with $|supp (U) \cup supp (V)|\leq t$, then $|\langle U,(I-\Phi^{\ast}(\Phi\Phi^{\ast})^{-1}\Phi)V\rangle|\leq\gamma_{t}\|U\|_{F}\|V\|_{F}$. Suppose $|R \cup supp (V)|\leq t$, then $\|[(I-\Phi^{\ast}(\Phi\Phi^{\ast})^{-1}\Phi)V]_{(R)}\|_{F}\leq\gamma_{t}\|V\|_{F}$.
\end{remark}

\begin{lemma}\label{rema}
For $E\in\mathbb{C}^{M\times L}$, $\|[\Phi^{\ast}(\Phi\Phi^{\ast})^{-1}E]_{(T)}\|_{F}\leq\sqrt{1+\theta_{t}}\|E\|_{F}$, where $\theta_{t}=\delta_{t}((\Phi\Phi^{\ast})^{-1}\Phi)$ and $\delta_{t}((\Phi\Phi^{\ast})^{-1}\Phi)$ is RIC of matrix $(\Phi\Phi^{\ast})^{-1}\Phi$.
\end{lemma}
\begin{proof}
\begin{equation*}
 \begin{array}{l}
\|[\Phi^{\ast}(\Phi\Phi^{\ast})^{-1}E]_{(T)}\|_{F}^{2}=\langle \Phi^{\ast}(\Phi\Phi^{\ast})^{-1}E, \bm{H}_{T}(\Phi^{\ast}(\Phi\Phi^{\ast})^{-1}E) \rangle\\
=\langle E, (\Phi\Phi^{\ast})^{-1}\Phi\bm{H}_{T}(\Phi^{\ast}(\Phi\Phi^{\ast})^{-1}E)\rangle\\
\leq\|E\|_{F}\sqrt{1+\theta_{t}}\|[\Phi^{\ast}(\Phi\Phi^{\ast})^{-1}E]_{(T)}\|_{F}.
 \end{array}
\end{equation*}
Applying {\em Definition \ref{RIP}} to the matrix $\Phi^{\ast}(\Phi\Phi^{\ast})^{-1}$ obtains the last step.
Hence, for all $E\in\mathbb{C}^{M\times L}$, we have $\|[\Phi^{\ast}(\Phi\Phi^{\ast})^{-1}E]_{(T)}\|_{F}\leq\sqrt{1+\theta_{t}}\|E\|_{F}$.
\end{proof}

\begin{lemma}\label{Lemma2.11}
 Let $Y=\Phi X+E$, where $X\in\mathbb{C}^{N\times L}$ is $s$ row-sparse with $S=$\text{supp}$(X)$ and $E\in\mathbb{C}^{M\times L}$ is the measurement error. If $\widetilde{W}\in\mathbb{C}^{N\times L}$ is $\widetilde{s}$ row-sparse, $\widetilde{X}=\widetilde{W}+\Phi^{\ast}(\Phi\Phi^\ast)^{-1}(Y-\Phi \widetilde{W})$, $\widetilde{Q}=orth(\widetilde{X})$, and $T$ is an index set of $t\geq s$ largest $\|\widetilde{Q}_{i\cdot}\|_{2}$, then
\begin{equation*}\label{h}
 \begin{array}{l}
\|X_{(T^{c})}\|_{F}\leq\sqrt{2}\alpha(\gamma_{s+\widetilde{s}+t}\|X-\widetilde{W}\|_{F}+\sqrt{1+\theta_{t+s}}\|E\|_{F}),
 \end{array}
\end{equation*}
where $\theta_{t+s}(\Phi)=\delta_{t+s}((\Phi\Phi^\ast)^{-1}\Phi)$.
\end{lemma}

\begin{proof}
Since rank$(Y)=L$ and $Y=\Phi\widetilde{X}$, it is obvious that rank$(\widetilde{X})=L$. Consequently, the singular value decomposition of $\widetilde{X}$ can be denoted as
$\widetilde{X}=\widetilde{U}_{\ell}\widetilde{\Sigma}_{(\ell)}\widetilde{V}^{\ast}$, where $\widetilde{U}_{\ell}$ is the first $L$ columns of $\widetilde{U}$ and $\widetilde{\Sigma}_{(\ell)}$ denotes the first $L$ rows of $\widetilde{\Sigma}$. Since $\widetilde{U}_{\ell}$ can be regarded as an orthonormal basis for the range of $\widetilde{X}$, without loss of generality, let $\widetilde{Q}=\widetilde{U}_{\ell}$, we have
\begin{equation*}
 \begin{array}{l}
\|[\widetilde{X}\widetilde{V}\widetilde{\Sigma}_{(\ell)}^{-1}]_{(T)}\|_{F}\geq\|[\widetilde{X}\widetilde{V}\widetilde{\Sigma}_{(\ell)}^{-1}]_{(S)}\|_{F}.\\
 \end{array}
\end{equation*}
It then follows that
\begin{equation*}
 \begin{array}{l}
\widetilde{\sigma}_{\min}^{-1}\|\widetilde{X}_{(T)}\|_{F}\geq\widetilde{\sigma}_{\max}^{-1}\|\widetilde{X}_{(S)}\|_{F},\\
 \end{array}
\end{equation*}
where $\widetilde{\sigma}_{\min}$ and $\widetilde{\sigma}_{\max}$ denote the smallest and the largest singular value of $\widetilde{\Sigma}_{(\ell)}$.
Eliminating the common terms over $T\bigcap S$, we obtain
\begin{equation*}
 \begin{array}{l}
\widetilde{\sigma}_{\min}^{-1}\|[\widetilde{W}+\Phi^{\ast}(\Phi\Phi^\ast)^{-1}(Y-\Phi \widetilde{W})]_{(T\setminus S)}\|_{F}\geq\widetilde{\sigma}_{\max}^{-1}\|[\widetilde{W}+\Phi^{\ast}(\Phi\Phi^\ast)^{-1}(Y-\Phi \widetilde{W})]_{(S\setminus T)}\|_{F}.\\
 \end{array}
\end{equation*}
For the left hand,
\begin{equation*}
 \begin{array}{l}
\widetilde{\sigma}_{\min}^{-1}\|[\widetilde{W}+\Phi^{\ast}(\Phi\Phi^\ast)^{-1}(Y-\Phi \widetilde{W})]_{(T\setminus S)}\|_{F}\\
=\widetilde{\sigma}_{\min}^{-1}\|[\widetilde{W}-X+\Phi^{\ast}(\Phi\Phi^\ast)^{-1}(\Phi X+E-\Phi \widetilde{W})]_{(T\setminus S)}\|_{F}\\
=\widetilde{\sigma}_{\min}^{-1}
\|[(I-\Phi^{\ast}(\Phi\Phi^\ast)^{-1}\Phi)(\widetilde{W}-X)+\Phi^{\ast}(\Phi\Phi^\ast)^{-1}E]_{(T\setminus S)}\|_{F}.\\
 \end{array}
\end{equation*}
The right hand satisfies
\begin{equation*}
 \begin{array}{l}
\widetilde{\sigma}_{\max}^{-1}\|[\widetilde{W}+\Phi^{\ast}(\Phi\Phi^\ast)^{-1}(Y-\Phi \widetilde{W})]_{(S\setminus T)}\|_{F}\\
=\widetilde{\sigma}_{\max}^{-1}
\|[\widetilde{W}+\Phi^{\ast}(\Phi\Phi^\ast)^{-1}(\Phi X+E-\Phi \widetilde{W})+X-X]_{(S\setminus T)}\|_{F}\\
\geq\widetilde{\sigma}_{\max}^{-1}\|X_{(S\setminus T)}\|_{F}-\widetilde{\sigma}_{\max}^{-1}
\|[(I-\Phi^{\ast}(\Phi\Phi^\ast)^{-1}\Phi)(\widetilde{W}-X)+\Phi^{\ast}(\Phi\Phi^\ast)^{-1}E]_{(S\setminus T)}\|_{F}.\\
 \end{array}
\end{equation*}
Therefore, we obtain
\begin{equation*}
 \begin{array}{l}
\widetilde{\sigma}_{\max}^{-1}\|X_{(S\setminus T)}\|_{F}\\
\leq \widetilde{\sigma}_{\max}^{-1}
\|[(I-\Phi^{\ast}(\Phi\Phi^\ast)^{-1}\Phi)(\widetilde{W}-X)+\Phi^{\ast}(\Phi\Phi^\ast)^{-1}E]_{(S\setminus T)}\|_{F}\\
+\widetilde{\sigma}_{\min}^{-1}
\|[(I-\Phi^{\ast}(\Phi\Phi^\ast)^{-1}\Phi)(\widetilde{W}-X)+\Phi^{\ast}(\Phi\Phi^\ast)^{-1}E]_{(T\setminus S)}\|_{F}\\
\leq\sqrt{2}\widetilde{\sigma}_{\min}^{-1}
\|[(I-\Phi^{\ast}(\Phi\Phi^\ast)^{-1}\Phi)(\widetilde{W}-X)+\Phi^{\ast}(\Phi\Phi^\ast)^{-1}E]_{(T\triangle S)}\|_{F}\\
\leq\sqrt{2}\widetilde{\sigma}_{\min}^{-1}\|[(I-\Phi^{\ast}(\Phi\Phi^\ast)^{-1}\Phi)(\widetilde{W}-X)]_{(T\triangle S)}\|_{F}\sqrt{2}\widetilde{\sigma}_{\min}^{-1}\|[\Phi^{\ast}(\Phi\Phi^\ast)^{-1}E]_{(T\triangle S)}\|_{F}\\
\leq\sqrt{2}\widetilde{\sigma}_{\min}^{-1}(\gamma_{s+\widetilde{s}+t}\|X-\widetilde{W}\|_{F}+\sqrt{1+\theta_{t+s}}\|E\|_{F}).
 \end{array}
\end{equation*}
The last step is due to  {\em Remark \ref{remark11}} and {\em Lemma \ref{rema}}. In view of {\em Definition \ref{feasible}}, we derive
\begin{equation*}
 \begin{array}{l}
\|X_{(S\setminus T)}\|_{F}\leq\sqrt{2}\alpha(\gamma_{s+\widetilde{s}+t}\|X-\widetilde{W}\|_{F}+\sqrt{1+\theta_{t+s}}\|E\|_{F}).
 \end{array}
\end{equation*}
\end{proof}

\begin{lemma}\label{lemm6}
Let $Y=\Phi X+E$, where $X\in\mathbb{C}^{N\times L}$ is $s$ row-sparse signal matrix, and $E\in\mathbb{C}^{M\times L}$ is the measurement error.  Let $S=$supp$(X)$ be the index set of the $s$ sparse rows of $X$.  Denote by $\widetilde{Q}=$orth$(\widetilde{X})$ the orthogonal basis of the row-space of $X$, and $T$ the index set of $t\geq s$ largest values of $\|\widetilde{Q}_{i\cdot}\|_{2}$. If $\overline{W}$ is the feedback of $\widetilde{X}$ given by $\overline{W}_{(T)}=\widetilde{X}_{(T)}+(\Phi_{T}^{\ast}\Phi_{T})^{-1}\Phi_{T}^{\ast}\Phi_{T^{c}}\widetilde{X}_{(T^{c})}$ and $\overline{W}_{(T^{c})}=0$, then
\begin{equation*}\label{hh}
 \begin{array}{l}
\|(X-\overline{W})\|_{F}\leq\frac{\|X_{(T^{c})}\|_{F}}{\sqrt{1-\delta_{s+t}^{2}}}+\frac{\sqrt{1+\delta_{t}}\|E\|_{F}}{1-\delta_{s+t}}.
 \end{array}
\end{equation*}
\end{lemma}

\begin{proof}
For any $Z\in\mathbb{C}^{N\times L}$ supported on $T$,
\begin{equation*}
 \begin{array}{l}
\langle \Phi \overline{W}-Y,\Phi Z\rangle\\
=\langle \Phi_{T}\widetilde{X}_{(T)}+\Phi_{T}( \Phi_{T}^{\ast} \Phi_{T})^{-1}\Phi_{T}^{\ast} \Phi_{T^{c}}\widetilde{X}_{(T^{c})}-Y,\Phi_{T}Z_{(T)}\rangle\\
=\langle \Phi_{T}^{\ast}(\Phi_{T}\widetilde{X}_{(T)}+\Phi_{T^{c}}\widetilde{X}_{(T^{c})}-Y),Z_{(T)}\rangle\\
=\langle \Phi_{T}^{\ast}(\Phi \widetilde{X}-Y),Z_{(T)}\rangle\\
=0.\\
 \end{array}
\end{equation*}
The last step is due to the feasibility of $\widetilde{X}$.
The inner product can also be written as
$\langle \Phi\overline{W}-Y,\Phi Z\rangle=\langle (\Phi\overline{W}-\Phi X-E),\Phi Z\rangle=0$. Therefore,
$\langle (\overline{W}-X),\Phi^{\ast}\Phi Z\rangle=\langle E,\Phi Z\rangle,~\forall~ Z\in\mathbb{C}^{N\times L}$ supported on $T$. Since $(\overline{W}-X)_{T}$ is supported on $T$, one has $\langle (\overline{W}-X),\Phi^{\ast}\Phi_{T}(\overline{W}-X)_{(T)}\rangle=\langle E,\Phi_{T}(\overline{W}-X)_{(T)}\rangle.$

\noindent Consequently,
\begin{equation*}
 \begin{array}{l}
\|(\overline{W}-X)_{(T)}\|_{F}^{2}=\langle(\overline{W}-X),\bm{H}_{T}(\overline{W}-X)\rangle\\
=|\langle (X-\overline{W}),(I-\Phi^{\ast}\Phi)\bm{H}_{T}(X-\overline{W})\rangle+|\langle E,\Phi\bm{H}_{T}(X-\overline{W})\rangle|\\
\leq\delta_{s+t}\|X-\overline{W}\|_{F}\|(X-\overline{W})_{(T)}\|_{F}+\sqrt{1+\delta_{t}}\|E\|_{F}\|(X-\overline{W})_{(T)}\|_{F}.\\
 \end{array}
\end{equation*}
The last step is due to {\em Lemma \ref{lemma2.7}} and {\em Definition \ref{RIP}}.
\noindent We can obtain
\begin{equation*}
 \begin{array}{l}
\|(X-\overline{W})_{(T)}\|_{F}\leq\delta_{s+t}\|X-\overline{W}\|_{F}+\sqrt{1+\delta_{t}}\|E\|_{F}.
 \end{array}
\end{equation*}
It then follows that
\begin{equation*}
 \begin{array}{l}
\|(X-\overline{W})\|_{F}^{2}=\|(X-\overline{W})_{(T)}\|_{F}^{2}+\|(X-\overline{W})_{(T^{c})}\|_{F}^{2}\\
\leq(\delta_{s+t}\|X-\overline{W}\|_{F}+\sqrt{1+\delta_{t}}\|E\|_{F})^{2}+\|X_{(T^{c})}\|_{F}^{2}.\\
 \end{array}
\end{equation*}
This in turn implies $p(\|X-\widetilde{W}\|_{F})\leq0$, where $p(\cdot)$  is a quadratic polynomial, defined by
\begin{equation*}
 \begin{array}{l}
p(x)=(1-\delta_{s+t}^{2})x^{2}-2\delta_{s+t}\sqrt{1+\delta_{t}}\|E\|_{F}x-(1+\delta_{t})\|E\|_{F}^{2}-\|X_{(T^{c})}\|_{F}^{2}.
 \end{array}
\end{equation*}
\noindent Since $(1-\delta_{s+t}^{2})\geq0$, it means that $\|(X-\overline{W})\|_{F}$ is smaller than the largest root of $p(\cdot)$
\begin{equation*}
 \begin{array}{l}
\|(X-\overline{W})\|_{F}\leq\frac{\delta_{s+t}\sqrt{1+\delta_{t}}\|E\|_{F}+\sqrt{(1+\delta_{t})\|E\|_{F}^{2}+({1-\delta_{s+t}^{2})\|X_{(T^{c})}\|_{F}^{2}}}}{1-\delta_{s+t}^{2}}\\
\leq\frac{\|X_{(T^{c})}\|_{F}}{\sqrt{1-\delta_{s+t}^{2}}}+\frac{\sqrt{1+\delta_{t}}\|E\|_{F}}{1-\delta_{s+t}}.\\
 \end{array}
\end{equation*}
\end{proof}

\begin{theorem} \label{the1}
Let $Y=\Phi X+E$, where  $X$ is the $s$ row-sparse signal matrix.
Then the sequence $\{W^{k}\}$ produced by OSNST+HT+$f$-FB satisfies
\begin{equation*}
 \begin{array}{l}
\|(X-W^{k})\|_{F}\leq\rho_{s+f(k)+f(k-1)}^{k}\|X-W^{0}\|_{F}+\frac{\kappa_{s+f(k)+f(k-1)}(1-\rho_{s+f(k)+f(k-1)}^{k})}{1-\rho_{s+f(k)+f(k-1)}}\|E\|_{F},\\
 \end{array}
\end{equation*}
where $\rho_{\ell}=\sqrt{\frac{2\alpha^{2}\gamma_{\ell}^{2}}{1-\delta_{\ell}^{2}}}$ and $\kappa_{\ell}=(\frac{\sqrt{1+\delta_{\ell}}}{1-\delta_{\ell}}+
\frac{\sqrt{2\alpha^{2}(1+\theta_{\ell})}}{\sqrt{1-\delta_{\ell}^{2}}})$.
\end{theorem}

\begin{proof}
 Applying {\em Lemma \ref{Lemma2.11}} to $\widetilde{W}=W^{k-1}$ and $T=T_{k}$ gives
\begin{equation*}
 \begin{array}{l}
\|X_{(T^{c}_{k})}\|_{F}\leq
\sqrt{2}\alpha(\gamma_{s+f(k-1)+f(k)}\|X-W^{k-1}\|_{F}+\sqrt{1+\theta_{s+f(k)}}\|E\|_{F}),
 \end{array}
\end{equation*}
\noindent and setting $\overline{W}=W^{k}$ and $T=T_{k}$ in {\em Lemma \ref{lemm6}} obtains
\begin{equation*}
\begin{array}{l}
\|(X-W^{k})\|_{F}\leq\frac{\|X_{(T^{c}_{k})}\|_{F}}{\sqrt{(1-\delta_{s+f(k)}^{2})}}+\frac{\sqrt{1+\delta_{f(k)}}\|E\|_{F}}{1-\delta_{s+f(k)}}.
 \end{array}
\end{equation*}
\noindent Combining these two inequalities, we have
\begin{equation*}
\begin{array}{l}
\|(X-W^{k})\|_{F}\leq\sqrt{\frac{2\alpha^{2}\gamma_{s+f(k)+f(k-1)}^{2}}{(1-\delta_{s+f(k)}^{2})}}\|X-W^{k-1}\|_{F}+(\frac{\sqrt{1+\delta_{f(k)}}}{1-\delta_{s+f(k)}}+
\frac{\sqrt{2\alpha^{2}(1+\theta_{s+f(k)})}}{\sqrt{1-\delta_{s+f(k)}^{2}}})\|E\|_{F}.
 \end{array}
\end{equation*}
Since $\delta_{\ell} $ and $\gamma_{\ell} $ are all non-decreasing \cite{Candes}, $\rho_{\ell}$ and $\kappa_{\ell}$ are also all non-decreasing as $\ell$ increases for all integer $\ell$. Note that $f(\ell)$ is also a nondecreasing function, it then follows that
\begin{equation*}
\begin{array}{l}
\|(X-W^{k})\|_{F}\leq\rho_{s+f(k)+f(k-1)}^{k}\|X-W^{0}\|_{F}+\frac{\kappa_{s+f(k)+f(k-1)}(1-\rho_{s+f(k)+f(k-1)}^{k})}{1-\rho_{s+f(k)+f(k-1)}}\|E\|_{F}.
 \end{array}
\end{equation*}
\end{proof}

Consequently, if the RIP and the P-RIP of the matrix $\Phi$ obeys $2\alpha^{2}\gamma_{s+f(k)+f(k-1)}^{2}+\delta_{s+f(k)+f(k-1)}^{2}<1$, the OSNST+HT+FB algorithm is guaranteed to converge.

\begin{figure}[th]
\centering
\begin{tabular}{c@{\hskip -0.9cm}c}
\resizebox {1.85in}{1.5in} {\includegraphics{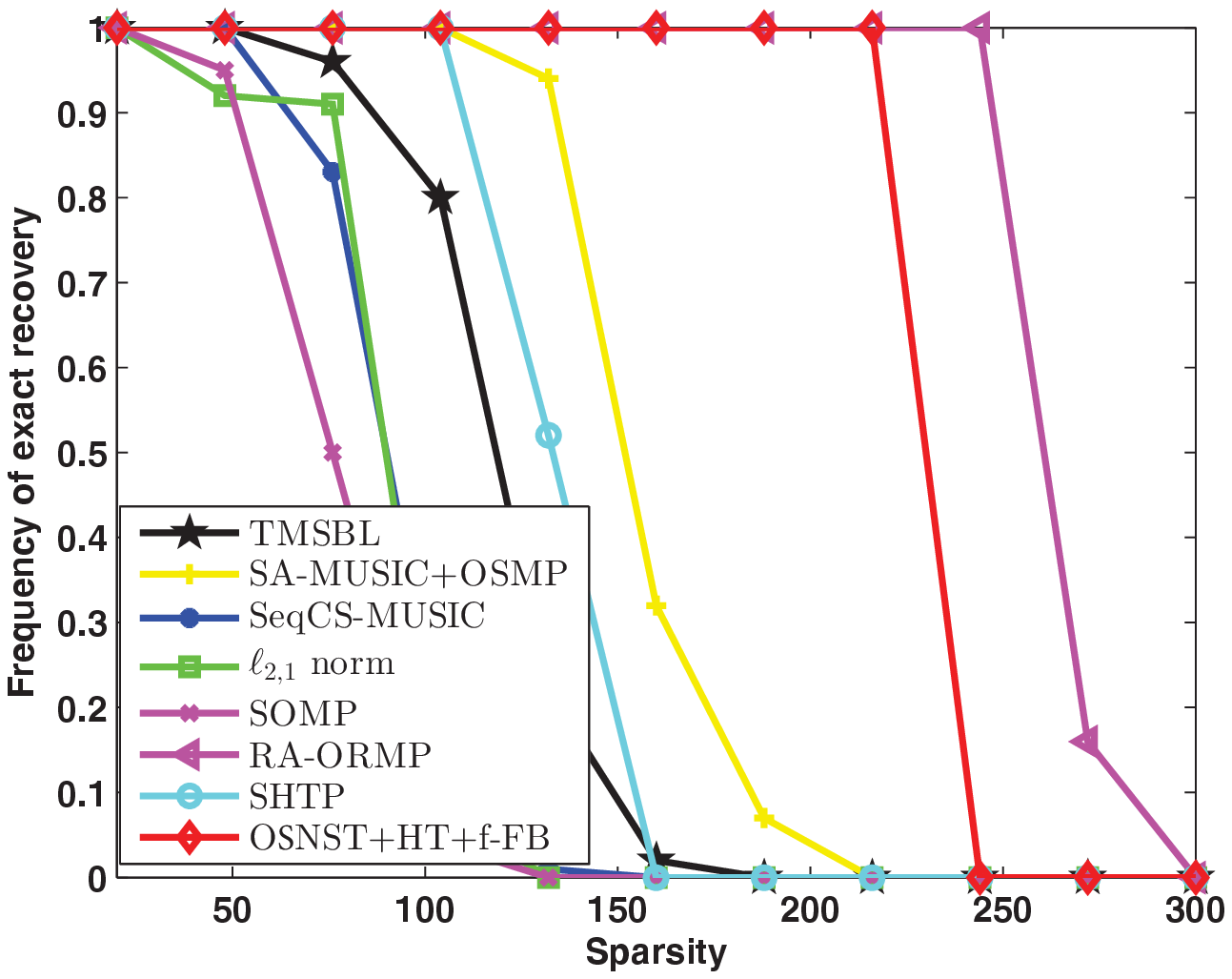}}
\quad & \quad
\resizebox {1.85in}{1.5in} {\includegraphics{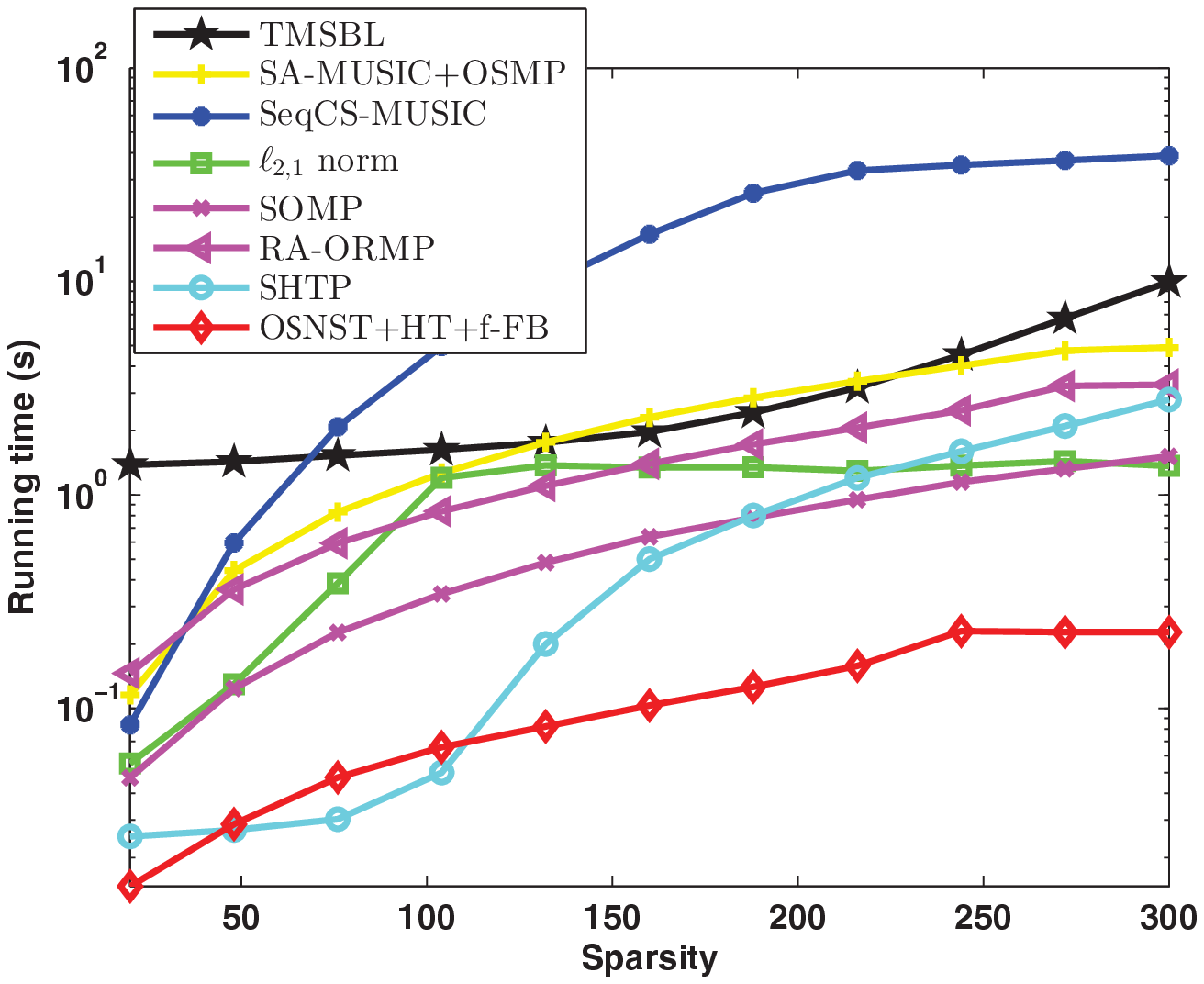}}
\end{tabular}
\vskip -0.3cm
\caption{\small  Left: Frequency of exact recovery as a function of sparsity; right: running time as a function of sparsity.}
\label{figure:Micro_Doppler}
\end{figure}

\section{EXPERIMENTS}
\label{}
In this experiment, the measurement matrix $\Phi$ is an $300\times1000$ Gaussian random matrix and the number of snapshots is $10$. To model the temporal correlation of MMV problem, we employ an autoregressive process of order $1$, AR(1). As a result, the $j$-th snapshot $X_{\cdot j}$ is generated according to the model
\begin{equation*}
\begin{array}{l}
X_{\cdot j}=\beta X_{\cdot (j-1)}+(1-\beta)\epsilon_{j},
 \end{array}
\end{equation*}
where $\beta$ is the AR model parameter controlling the temporal correlation and $\epsilon_{j}$ is the level of white Gaussian perturbation. The support of a sparse signal is also chosen randomly and the nonzero entries of Gaussian sparse signals are drawn independently from the Gaussian distribution with zero mean and unit variance. A successful recovery is recorded when $\|X-\widehat{X}\|_{F}/\|X\|_{F}\leq10^{-4}$, where $X$ is the exact signal matrix and $\widehat{X}$ denotes the recovered signal. Each experiment is tested for $100$ (random) trials. A matlab implementation of the proposed algorithm is also available at

 \url{https://www.dropbox.com/s/2avudk770m4c6rz/OSNST%2BHT%2Bf-FB.zip?dl=0}.

We first study the mechanisms of $f$-feedback by introducing six particular index selection functions: $f(x)=x$, $f(x)=3x$, $f(x)=6x$, $f(x)=9x$, $f(x)=12x$ and $f(x)=x^{2}$. As discussed, higher
critical sparsity represents better empirical recovery performance. Figure 1 shows the frequency of exact recovery and the running time as functions of the sparsity levels $s$. As shown, linear functions with modest gradients present similar performance, which is better than the quadratic function $f(x)=x^{2}$. In addition, one can accelerate the convergence of the class of OSNST+HT+$f$-FB algorithms by adjusting the cardinality of indices per iteration.

Also presented are comparisons among our OSNST+HT+$f$-FB and state-of-the-art techniques such as SOMP \cite{Tropp1}, $\ell_{2,1}$ norm \cite{Berg}, SHTP  \cite{Foucart,Blanchard}, RA-ORMP \cite{Eldar}, TMSBL \cite{Zhangqq}, SA-MUSIC+OSMP \cite{KLee}, SeqCS-MUSIC \cite{MKim,Kimdd} in terms of frequency of exact recovery and running time. In this experiment, we adopt a modest setting $f(x)=6x$, which can be applied to other applications. In Figure 2, experimental results  show that OSNST+HT+$f$-FB still delivers reasonable performance better than that of  SOMP, $\ell_{2,1}$ norm, SHTP, TMSBL, SA-MUSIC+OSMP, and SeqCS-MUSIC, though slightly under-performs that of  RA-ORMP. For the execution-time comparison, our algorithm achieves the best performance. Numerical experiments show that our algorithm has a clearly advantageous balance of efficiency, adaptivity and accuracy compared with other state-of-the-art algorithms.

\ifCLASSOPTIONcaptionsoff
  \newpage
\fi

%




\end{document}